\documentclass[aps,prx,showpacs,twocolumn,amsmath,amsfonts,amssymb,superscriptaddress,tightenlines]{revtex4-1}
\usepackage{graphicx}
\usepackage{bm}
\usepackage{enumerate}
\usepackage{hyperref}

\usepackage{floatrow}
\usepackage{braket}
\newfloatcommand{capbtabbox}{table}[][\FBwidth]
\usepackage{paralist}

\usepackage{algorithm}
\usepackage{algpseudocode}

\usepackage{amsthm}
\newtheorem{theorem}{Theorem}[section]

\usepackage{blkarray} 

\usepackage{xcolor}
\newcommand{\highlight}[2]{%
  \colorbox{#1}{$\displaystyle#2$}}
\definecolor{Set1-A}{RGB}{228,26,28}
\definecolor{Set1-B}{RGB}{55,126,184}
\definecolor{Set1-C}{RGB}{77,175,74}

\begin{document}

\title{Persistent homology of quantum entanglement}

\begin{abstract}
Structure in quantum entanglement entropy is often leveraged to focus on a small corner of the exponentially large Hilbert space and efficiently parameterize the problem of finding ground states. A typical example is the use of matrix product states for local and gapped Hamiltonians. We study the structure of entanglement entropy using persistent homology, a relatively new method from the field of topological data analysis. The inverse quantum mutual information between pairs of sites is used as a distance metric to form a filtered simplicial complex. Both ground states and excited states of common spin models are analyzed as an example. Furthermore, the effect of homology with different coefficients and boundary conditions is also explored. Beyond these basic examples, we also discuss the promising future applications of this modern computational approach, including its connection to the question of how spacetime could emerge from entanglement.
\end{abstract}

\author{Bart Olsthoorn} 
\email{bartol@kth.se}
\affiliation{Nordita, KTH Royal Institute of Technology and Stockholm University, Hannes Alfvéns väg 12, SE-106 91 Stockholm, Sweden}
\date{\today}
\maketitle

\section{Introduction}
Problems in quantum physics are often difficult to solve due to an exponentially large Hilbert space (e.g. $d=2^N$ for $N$ spins), which limits exact diagonalization (ED) to small systems only. However, due to the remarkable entanglement scaling properties found in many physical systems, it is often possible to focus only on a small corner of the Hilbert space. For example, the problem of finding a ground state of a one-dimensional local and gapped Hamiltonian can be parameterized efficiently with matrix product states (MPS) \cite{Hastings2007}. Beyond this class of Hamiltonians, the entanglement properties have also proven useful in projected entangled-pair states (PEPS) \cite{Verstraete2006} and the Multiscale Entanglement Renormalizaton Ansatz (MERA) \cite{PhysRevLett.101.110501}.

Phases of matter can differ in their entanglement scaling properties, and the critical phase transition is typically a point of special interest. The area law states that the entanglement entropy only depends on the surface between two subregions of the system \cite{Eisert2010,PhysRevLett.71.666}. The transverse field Ising model (TFIM) exhibits a quantum phase transition from the ferromagnet with area law entanglement that to a paramagnet that also follows an area law. At the quantum critical point (QCP), the entanglement entropy diverges logarithmically with the subsystem size (volume law) \cite{Holzhey1994,Calabrese2004,Dutta2015}, thus allowing the identification of the phase transition. Another example of a phase transition that is studied through entanglement properties is the many-body localization (MBL) transition \cite{Bauer_2013,Herviou2019}. In general, a randomly selected state in Hilbert space typically follows a volume law.

Considering the daunting complexity of Hilbert space, various large scale computational techniques are employed: phase transitions and order parameters are now often analyzed and constructed with machine learning \cite{vanNieuwenburg2017,Zvyagintseva2022}. 

Here we outline the proposal to use persistent homology to identify entanglement structures and quantum phase transitions. Persistent homology analysis has been used to identify phase transitions in classical spin models \cite{Donato2016,Tran2021,Olsthoorn2020,Cole2021,Sale_2022,PhysRevB.106.085111} and some quantum systems \cite{Tran2021,Spitz2021,He2022-vj}. These examples also illustrate the evolving views on the importance and unique role that entanglement plays as an indicator of qualitative changes in quantum systems. Our approach focuses on the objects that are derivative of quantum states (simplicial complexes) and function as a key descriptor of global properties of the system. The main technique used in this paper is persistent homology, a relatively new method from the field of topological data analysis (TDA) that computes the shapes present in data \cite{10.5555/795666.796607,10.5555/3115476.3115792,VRobins1999,Frosini_sizetheory}. In short, TDA often involves turning a point cloud in data space into a filtered simplicial complex where the filtration is done with a chosen length scale parameter. At each filtration stage, the homology groups are computed and compared. Homology group elements that persist over many stages lead to persistent homology and can be visualized in a barcode or persistence diagram. For more details and a practical introduction to persistent homology, see ref~\cite{Otter2017}.

\begin{figure}[t]
    \centering
    \includegraphics{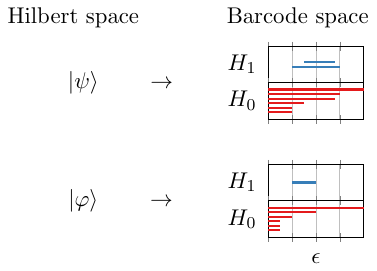}
    \caption{A given state $\psi$ in Hilbert space is mapped to a barcode that describes the topological structure of the entanglement of its subsystems. The proximity parameter $\epsilon$ is computed from the quantum mutual information, and the homology groups indicate features in the emergent geometry.}
    \label{fig:schematic_hilbert_barcodes}
\end{figure}

In this work, we study the geometric and topological features of the entanglement entropy with persistent homology. The logic of the approach is schematically shown in Fig.~\ref{fig:schematic_hilbert_barcodes}. A state $\ket{\psi}$ in Hilbert space is converted to a corresponding barcode that efficiently describes its entanglement structure. The barcode is the result of applying persistent homology to a distance matrix that quantifies the entanglement between all subsystems. Using this approach we discuss the changes in barcode space comparing ground states, excited states and phase transitions. As an example, we focus on two common quantum spin chains in a transverse field: Ising model and XXZ model. The former is the simplest example of a quantum phase transition, and the latter is a model commonly used to study many-body localization \cite{nidari2008,Herviou2019}.

The remainder of this paper is structured as follows. In Sec.~\ref{sec:outline} we outline how persistent homology captures the entanglement structure of a quantum state. An introduction to simplicial homology, persistence and our distance metric are presented in Sec.~\ref{sec:sh}, \ref{sec:ph} and \ref{sec:distance_metric}. A demonstration is carried out for two different spin models in Sec.~\ref{sec:ex_1} and \ref{sec:ex_2}, followed by a discussion in Sec.~\ref{sec:discussion}. We conclude in Sec.~\ref{sec:conclusion}.

\section{Entanglement and persistent homology}
\label{sec:outline}
The foundation of our work is the description of the entanglement structure of a quantum state using persistent homology. Given a state $\ket{\psi}$, we divide the quantum system into $N$ subsystems. For each pair of subsystems, the quantum mutual information $M_{ij}$ is computed. A structure appears when entangled subsystems (with large mutual information) form clusters. A distance metric is defined as the inverse mutual information between two subsystems. This brings strongly entangled subsystems close while non-entangled subsystems are far apart. Given a distance matrix between all subsystems, the computation of persistent homology groups makes it possible to study topological and geometrical features of the entanglement structure. This information can be plotted as a barcode, where long bars indicate topological features that persist over large length scales. The length scale refers to the entanglement through the defined distance metric. The barcode captures the global entanglement structure of the quantum state and gives a detailed view of the "phase portrait" of the system. How to compute this in practice is outlined in the next section.

The barcode is a summary of the topological and geometrical features of the entanglement structure. Each quantum state has a barcode and it can be used in a number of applications. In our work, we demonstrate that changes in the barcode reveal phase transitions. At the same time, the barcode reveals the length scale of the entanglement in the state which is important when developing suitable wave function ansatze. Finally, the barcode describes the emergent geometry of the entanglement, which could provide a starting point to study the emergence of spacetime from entanglement.

The barcodes corresponding to quantum states can be compared in a number of ways. Distance measures such as the Bottleneck distance \cite{Bubenik2015-ic} or the Wasserstein distance \cite{Villani_2009} operate directly on the complete barcodes rather than derived quantities. However, depending on the application, it can be sufficient to simply count the number of bars at a specific length scale (i.e. the Betti number $\beta_k$). We focus on Betti numbers in our work and demonstrate its sensitivity to quantum phase transitions. However, the computational algorithm outlined here is general and can be used in a variety of applications. 

\begin{figure*}
    \centering
    \includegraphics[width=\linewidth]{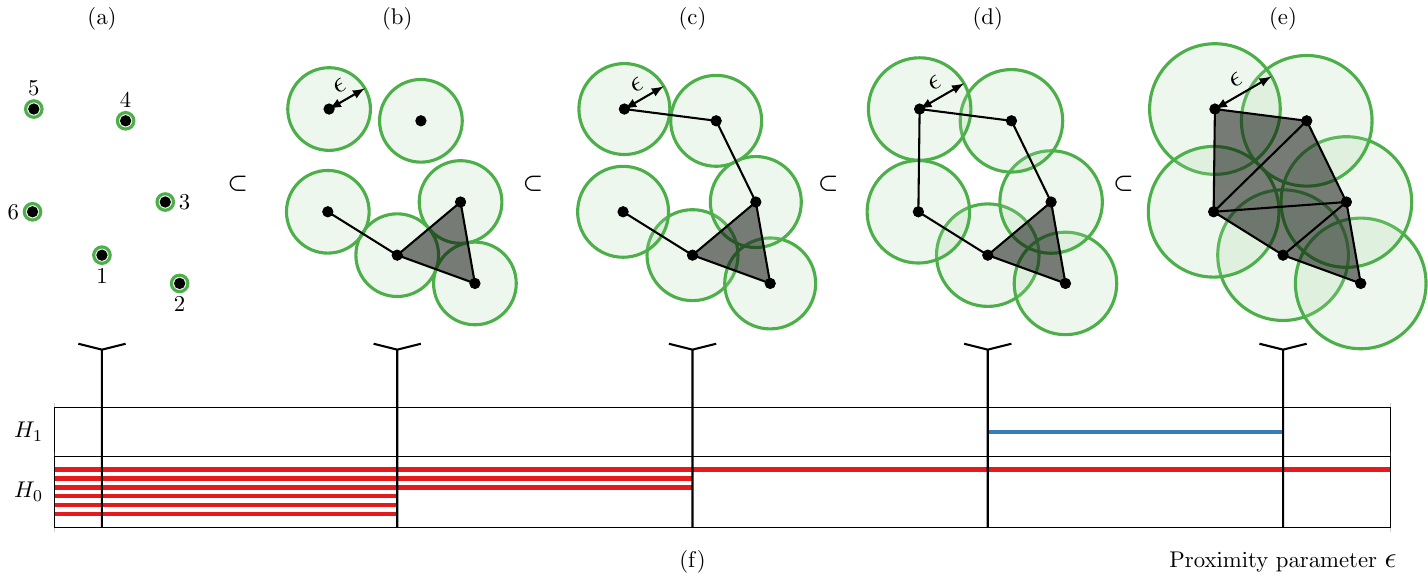}
    \caption{Vietoris-Rips complex of a point cloud with six data points (a) at increasing proximity parameter $\epsilon$ (b-e). The barcode (f) shows the persistent homology.}
    \label{fig:barcode_example}
\end{figure*}

\section{Simplicial homology}
\label{sec:sh}
Simplices are the building blocks for topological spaces and are convenient due to their combinatorial nature. A $k$-simplex consists of $k+1$ vertices, i.e. $S=\left[v_0 v_1 \dots v_k\right]$ (see Fig~\ref{fig:simplices}. Each simplex can be oriented in two ways, $-S$ or $+S$, however, this feature is not always used in persistent homology, as we will see. Combining simplices and connecting them forms a simplicial complex that represents a topological space. This section introduces the concepts of simplicial homology, namely, chains, boundaries, cycles and homology groups.

For each definition, we also study an example. The example simplicial complex $K$ corresponds to topological space that is a disk glued to a circle,
\begin{align}
    K&=\{\underbrace{[012]}_{\text{2-simplex}},\underbrace{[01],[12],[20],\dots}_{\text{1-simplices}},\underbrace{[0],[1],\dots}_{\text{0-simplices}}\}\nonumber\\
    &\cong\begin{gathered}\includegraphics{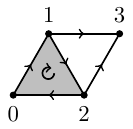}\end{gathered}
    \label{eq:K_example}
\end{align}
as becomes clear once we compute its homology.

\emph{Chains.}\ \ A $k$-chain is a sum of $k$-simplices in $K$. All the possible $k$-chains form a group. For the example in Equation~\ref{eq:K_example}, there are three chain groups: $C_0$, $C_1$ and $C_2$. The $1$-chains are for example:
\begin{align}
    C_1&=\{a_{01}[01]+a_{12}[12]+a_{20}[20]+a_{13}[13]+a_{23}[23]\}\nonumber\\
    &=\left<[01],[12],[20],[13],[23]\right>
    \label{eq:C_1_K}
\end{align}
where the $1$-simplices form a basis. An element in the chain group is denoted by $c\in C_k$, and this element represents a specific $k$-chain.

\emph{Boundaries.}\ \ The boundary operator $\partial$ for a $k$-simplex is defined as
\begin{equation}
    \partial S=\sum_{j=0}^k\left(-1\right)^j\left[v_0 v_1\dots \hat{v}_j \dots v_k\right]
    \label{eq:boundary_operator_Z}
\end{equation}
where $\hat{v}_j$ is removed from the sequence. In other words, it is a sum of the faces of the simplex, where the faces are $(k-1)$-simplices. The \emph{
Fundamental Lemma of Homology} \cite{books/daglib/0025666} states that applying the boundary twice always leads to zero, i.e. $\partial \partial S=0$. The boundary operator is typically applied to a chain group $C_k$, leading to a boundary group $B_k$. For Equation~\ref{eq:K_example}, the boundary group:
\begin{align}
    B_0&=\partial_1 C_{1}\nonumber\\
    &=\{
        \ a_{01}([1]-[0])+
        a_{12}([2]-[1])+\nonumber\\
        &\quad\quad a_{20}([0]-[2])+
        a_{13}([3]-[1])+\nonumber\\
        &\quad\quad a_{23}([3]-[2])
    \ \}\nonumber\\
    &=\{
    \ (-a_{01}+a_{20})[0]+(a_{01}-a_{12}-a_{13})[1]+\label{eq:B0_example}\nonumber\\
    &\quad\quad  \ (a_{12}-a_{20}-a_{23})[2]+(a_{13}+a_{23})[3]
    \ \}
\end{align}

\emph{Cycles.}\ \ A $k$-cycle is a $k$-chain with an empty boundary, i.e. $\partial_k c=0$. This group is denoted by $Z_k$, where $Z$ stands for the German \emph{zyklus} (cycle). By definition, the $k$-cycles are a subset of all the $k$-chains, i.e. $Z_i\subset C_i$. Mathematically, it corresponds to taking the kernel of the boundary operator, i.e. all the elements that are mapped to the identity element. In other words, the $k$-cycles group is,
\begin{equation}
Z_k=\text{Ker}(\partial_k)=\{c \in C_k\ |\ \partial(c) = 0\},
\end{equation}
because $0$ is the identity element. For Equation~\ref{eq:C_1_K}, it means the coefficients $a\in\mathbb{Z}$ are constrained by $\partial(c)=0$, narrowing it down to two generating elements for the 1-cycles group,
\begin{align}
    Z_1&=\left<\highlight{Set1-A!50}{[01]+[12]+[20]},\highlight{Set1-C!50}{[12]-[13]+[23]}\right>\nonumber\\
    &\cong \begin{gathered}\includegraphics{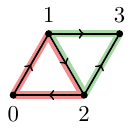}\end{gathered}.
\end{align}
Note that the larger cycle of $0-1-3-2$ can be constructed by a linear combination of this basis.

\emph{Homology groups.}\ \ The homology groups $H_k$ are defined as the quotient of cycles and boundaries,
\begin{equation}
    H_k=\frac{Z_k}{B_k}=\frac{\text{Ker}\  \partial_k}{\text{Im}\ \partial_{k+1}},
\end{equation}
where $Z_k$ are the $k$-cycles and $B_k$ the $k$-boundaries. This is sometimes referred to as ``cycles mod boundaries". 
The first homology group for our example contains one generating element,
\begin{align}
    H_1(K)&=\frac{\left<[01]+[12]+[20],[12]-[13]+[23]\right>}{\left<[12]-[02]+[01]\right>}\nonumber\\
    &\cong\left<[12]-[13]+[23]\right>\cong \mathbb{Z}, 
\end{align}
which captures the 1-dimensional hole (of the empty triangle $1-3-2$). The $1$-cycle between the points $0-1-2$ is ``modded out" because it is the boundary of the 2-simplex (filled triangle). A similar computation leads to the other homology groups. In summary, the homology groups for the example $K$ are $H_0(K)\cong \mathbb{Z}$, $H_1(K)\cong \mathbb{Z}$ and $H_k(K)\cong 0$ for all $k > 1$. In other words, the topological space $K$ contains one connected component (as described by $H_0$) and contains one 1-dimensional hole (as described by $H_1$).

In this example, the homology computation is performed over the integers because the coefficients are chosen to be $a\in \mathbb{Z}$, but most practical codes use booleans, i.e. $a\in \mathbb{Z}_2$. For more examples of homology computations and the effect of different coefficients, see Appendix~\ref{sec:non_orientable_torsion}.

\section{Persistent homology}
\label{sec:ph}

The persistent homology of a set of discrete data characterizes the shapes that are present in the data at different length scales. In practice, this means the simplicial homology computations described in the previous section with a method of \emph{filtration}. The filtration refers to the choice of a distance metric, and how this metric leads to a set of nested simplicial complexes. A simple example of six data points in the two-dimensional Euclidean plane is shown in Fig.~\ref{fig:barcode_example}. The chosen proximity parameter $\epsilon$ is, in this case, given by simply the Euclidean distance, and visualized by the growing disks. As soon as disks overlap, a new simplex is formed. When three points are connected, a 2-simplex (filled triangle) is formed. In general, for $k$ points, a ($k$-1)-simplex is formed. This particular filtration leads to a so-called Vietoris-Rips complex. The simplicial homology at every level of filtration is summarized in the barcode Fig.~\ref{fig:barcode_example} (f). More specifically, a generating element of a homology groups that spans multiple length scales is indicated by a bar with a start (birth) and end (death). The length of the bar is also referred to as the lifetime.

In this study, the discrete data are the quantum subsystems, and the distance metric is based on mutual information as a measure of entanglement. The distance matrix between all quantum subsystems $i$ and $j$, i.e. $D_{ij}$, forms a sequence of simplicial complexes, $K_1 \subset K_2 \subset \dots \subset K_n$. For each simplicial complex $K_i$, the homology groups $H_k$ are computed, where $k$ refers to $k$-dimensional simplices (see Fig.~\ref{fig:simplices}). The homology groups $H_k$ are computed for each $K_i$ and sequences of identical homology lead to so-called persistent homology. The sequence of homology groups are typically visualized in a persistence barcode or persistence diagram. Both show the same information, and reveal shapes present in the quantum mutual information. A barcode is a collection of line segments that represent the generating elements of a homology group that span multiple length scales. When two elements become homologous, the older one survives, as described by the \emph{elder rule} \cite{books/daglib/0025666}.

In practical applications, since the barcode itself is not a scalar quantity, it is common to either use scalar properties derived from the barcode or construct a scalar distance between two barcodes. For example, the Betti number is an integer,
\begin{equation}
    \beta_k=\text{rank}(H_k)
    \label{eq:betti_number}
\end{equation}
that counts the rank of the $k$th homology group. Another example is the lifetime of a specific bar in the barcode. Bars with a long lifetime indicate that a feature (like a 1-dimensional hole) persists over large range of length scale. Features with a short lifetime are sometimes considered to be noise in the input point cloud. However, these features can also indicate the curvature of the manifold that the point cloud was sampled from \cite{BubenikCurvature}.

There are many different codes that implement the computation of persistent homology given a distance matrix. An example of the algorithm is shown in Appendix~\ref{sec:ph_algorithm}. We use the GUDHI Python module for the construction of a Vietoris-Rips complex from the distance matrix, and the computation of the barcode \cite{gudhi:RipsComplex}.

In our method, each state $\ket{\psi}$ has a corresponding barcode and a set of homology groups. The barcode is a fingerprint of the entanglement structure of the state. Abrupt changes can indicate phase transitions, and barcodes can also differentiate between phases.

\begin{figure}[t]
    \centering
    \includegraphics[width=\linewidth]{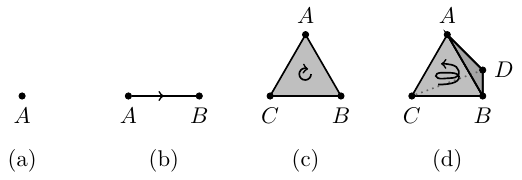}
    \caption{Here we show (a) $0$-simplex (point), (b) $1$-simplex (line segment), (c) $2$-simplex (filled triangle), (d) $3$-simplex (filled tetrahedron). Higher-dimensional simplices exist but are not shown here for simplicity. Simplices are glued together to form a simplicial complex and its topological properties are described by simplicial homology. Entangled quantum subsystems $A$, $B$, $C$ and $D$ are represented by 0-simplices. A $k$-simplex consists of $k+1$ subsystems.}
    \label{fig:simplices}
\end{figure}

\begin{figure}[b]
    \centering
    \includegraphics{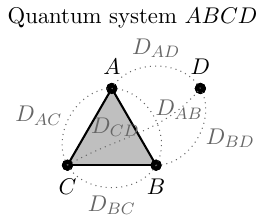}
    \caption{Example of a simplicial complex created from a quantum system comprising four subsystems: $A$, $B$, $C$ and $D$. The entanglement entropy between all subsystems is used to compute the distance $D_{ij}$ (Equation~\ref{eq:MI_distance}). This complex corresponds length scale $\epsilon$, where in this case subsystems $A$, $B$ and $C$ form a 2-simplex. Subsystem $D$ is not connected to the $ABC$ component, indicating that the distances between $D$ and $A$, $B$ and $C$ is larger than $\epsilon$. Here, the Betti numbers are $\beta_0=2$ (i.e. 2 connected components: $ABC$ and $D$) and $\beta_1=0$ (no 1-dimensional holes, since the 1-cycle $ABC$ is also the boundary of the 2-simplex $ABC$).}
    \label{fig:subsystem_complex}
\end{figure}

\section{Distance metric based on entanglement}
\label{sec:distance_metric}
The $N$ particles (or spins) form a discrete point cloud, where the distance metric between spins is the (additive) inverse of the quantum mutual information (MI). Mutual information is commonly used to identify clusters of entanglement, and as a probe for phase transitions (e.g. MBL and quantum phase transitions) \cite{Tran2021,PhysRevLett.118.016804,PhysRevA.97.042330,Herviou2019}. The mutual information between two sites $i$ and $j$ is defined as
\begin{equation}
    0 \leq M_{ij}=S_i + S_j - S_{ij} \leq 2\ln{2}.
    \label{eq:mutual_information}
\end{equation}
Here, the entanglement entropy $S_i$ and $S_{ij}$ refer to taking only sites $i$ and $i,j$ in subsystem $A$, respectively.
\begin{theorem}
Let $D_{ij}$ be the inverse of the MI between sites $i$ and $j$,
\begin{equation}
    2\ln{2} \geq D_{ij}=2\ln{2} - M_{ij} \geq 0.
    \label{eq:MI_distance}
\end{equation}
This is a distance metric that brings strongly entangled sites ($M\rightarrow 2\ln{2}$) close together ($D\rightarrow 0$), while non-entangled sites ($M\rightarrow 0$) are far apart ($D\rightarrow 2\ln{2}$).
\end{theorem}

\begin{proof}
The distance metric defined in Equation~\ref{eq:MI_distance} satisfies the axioms for a metric because it is symmetric and satisfies the triangle inequality. This triangle inequality relies on properties of quantum entropy, and we start by writing the inequality in terms of entropy, 
\begin{align}
D_{xy}&\leq D_{xz}+D_{zy}\\
2 S_z+S_{xy}-S_{xz}-S_{zy}&\leq 2\ln{2}.\nonumber
\end{align}
In order to show that the inequality holds, we first note that $S_{xz}$ and $S_{zy}$ are non-negative and the term $X=S_{xz}+S_{zy}$ can be replaced by $S_x+S_y$, which is guaranteed to be equal or smaller by the strong subadditivity of quantum entanglement (Equation~\ref{eq:S_strong_subadd}). Furthermore, the term $S_{xy}$ is replaced by an expression that is guaranteed to be larger, $S_{xy}\leq S_x+S_y$ (Equation~\ref{eq:S_triangle_ineq}). This leads to,
\begin{align}
    2 S_z+S_x+S_y-S_x-S_y&\leq 2\ln{2}\\
    2 S_z&\leq 2\ln{2},
\end{align}
which is true because $S_z$ has a maximum value of $\ln{2}$. Finally, the distance of point $i$ with itself, $D_{ii}$, is never evaluated when constructing homology, and we can consider it to be zero.
\end{proof}

Given a set of discrete data (quantum subsystems) and the distance matrix (inverse mutual information), we can now study its homology (see Fig.\ref{fig:subsystem_complex}).

\begin{figure}[b]
    \centering
    \includegraphics[]{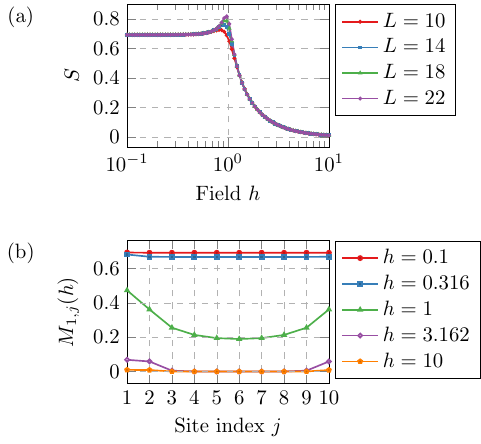}
    \caption{(a) Half-chain entanglement entropy $S$ as a function of field $h$ in the 1D TFIM. For small fields $h$ the entanglement is constant, $\ln(2)$. At the critical point of $h=1$, $S$ grows as $\ln(N)$. (b) Quantum mutual information between first site $1$ and site $j$ as a function of field $h$ (logarithmically spaced) in the 1D TFIM. The mutual information $M_{1j}$ is largest between nearby sites, indicating the local interactions of the Hamiltonian. Periodic boundary conditions are used and lead to increased $M_{1j}$ for large $j$.}
    \label{fig:1d_qmi_tfim}
\end{figure}

\begin{figure*}
    \centering
    \includegraphics{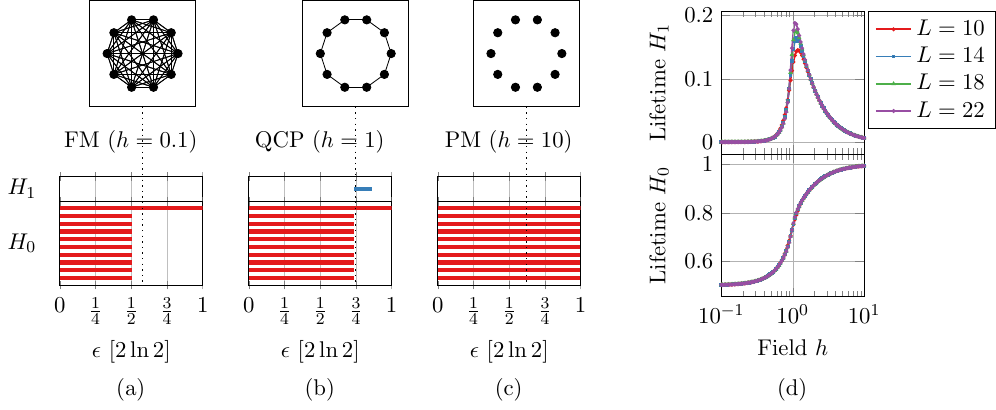}
    \caption{Barcodes for transverse-field Ising chain with $L=10$ sites at (a) $h=0.1$, (b) $1$ and (c) $10$. The proximity parameter $\epsilon$ refers to the length scale set by the distance matrix $D(i,j)$. The dotted lines and panels above show the simplicial complex for a given $\epsilon$. At the quantum critical point, the spin chain forms a cycle of entanglement with a 1-dimensional hole. (d) Lifetime of $H_1$ and $H_0$ bars. The persistence of the $1$-cycle at the QCP increases with system size $L$.}
    \label{fig:TFIM_field_barcodes}
\end{figure*}

\begin{figure}[b]
    \centering
    \includegraphics{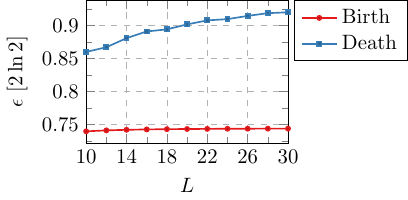}
    \caption{Birth and death of the $H_1$ persistence bar at the quantum critical point $h=1$ of the transverse-field Ising chain of varying size $L$. The birth value is set by the distance between neighboring sites and converges to $D_{01}/2\ln{2}\approx 0.744293$. The death value is set by the distance between two maximally separated sites and converges to $D_{0\infty}/2\ln{2}=1$ (see Appendix~\ref{sec:exact_MI_TFIM}).}
    \label{fig:ising_H1_b_d}
\end{figure}

\section{Example 1: Ising chain in transverse field}
\label{sec:ex_1}

To demonstrate our approach to persistent homology of quantum states, we study the simplest model of a quantum phase transition (QPT), the transverse field Ising model (TFIM) on a one-dimensional lattice. The Hamiltonian $\mathcal{H}$ of the transverse field Ising model (TFIM) is \cite{Dutta2015},
\begin{equation}
    \mathcal{H}=-\sum_{i}\sigma_i^x \sigma_{i+1}^x - h \sum_{i} \sigma_i^z,
\end{equation}
where $\sigma$ are the Pauli matrices. The exact diagonalization and entanglement entropy computation is implemented using the QuSpin package \cite{Weinberg_2017}. 
There is a QPT that can be measured by the magnetization $M=|1/N \sum_{i}\sigma_i^x|$. For $h=0$, the ferromagnetic ground state is the basis state with all spins pointing along the $x$ direction. For large $h$, the ground state is a quantum paramagnet, a superposition of all basis states.

For $h=0$, there is a degenerate ground state that is a product state of all spin up $\ket{\uparrow\uparrow\dots\uparrow}$ or all down $\ket{\downarrow\downarrow\dots\downarrow}$. For a small but finite $h$, these levels split with order $h^N$ and the ground state is a macroscopic superposition ($\ket{\uparrow\uparrow\dots\uparrow}+\ket{\downarrow\downarrow\dots\downarrow}$) with $\ln{2}$ half-chain entanglement entropy (see Fig.~\ref{fig:1d_qmi_tfim} (a)). The system undergoes a quantum phase transition at $h=1$, and the entanglement structure form a 1-cycle, as will be discussed. At $h=\infty$ the ground state is a a quantum paramagnet, i.e. superposition of all basis states $\frac{1}{\sqrt{2^N}}\sum_i^{2^N}\ket{i}$ and zero half-chain entanglement entropy.

The quantum ground state is converted into a barcode by computing the distance matrix $D$ between all sites. This constitutes constructing Vietoris-Rips complexes from this distance matrix and computing its homology groups. The barcodes are shown in Figure~\ref{fig:TFIM_field_barcodes}. At zero length scale $\epsilon=0$, the number of $H_0$ bars (i.e. $\beta_0$) indicates the number of quantum subsystems (spins in this case) and the dimensionality of the matrix $D$. At large length scale, at least one $H_0$ bar survives, indicating one large connected component, similar to the situation shown in the simple two-dimensional Euclidean example (Fig.~\ref{fig:barcode_example}).

Both the low ($h<1$) and high ($h>1$) field ground states have constant mutual information $M_{ij}$ for all pairs (see Fig.~\ref{fig:1d_qmi_tfim}). This causes all sites to pair up at the same length scale in the barcode, as indicated by the end of the $H_0$ bars, also referred to as the death of the feature. For the ferromagnetic (FM) phase, all the sites have the same pairwise distance and are connected at the length scale of $\epsilon= \ln{2}$ (see Fig.~\ref{fig:TFIM_field_barcodes} (a)). For the quantum paramagnet (PM), the spins are maximally separated, which is essentially a rescaled version of having points being infinitely separated in space, and each site survives as a single component up to the maximum length scale $\epsilon=1 \cdot 2\ln{2}$ (see Fig.~\ref{fig:TFIM_field_barcodes} (c)).

Around the quantum critical point ($h=1$), neighboring sites are more strongly entangled than next-nearest neighbors. Due to the symmetry of the Hamiltonian and the periodic boundary conditions, the distances are the same for each site, e.g. $D_{12}=D_{23}$. This means we have a 1-dimensional hole to form inside the ring of spins that persists over a finite length scale (see Fig.~\ref{fig:TFIM_field_barcodes} (b)). At slightly longer length scales, next-nearest neighbors are connected, after that next-next nearest neighbors, and so on. This process continues until furthest-separated spins on the chain are connected and the hole closes. This is a unique feature of the critical point due to the entanglement structure of the critical state, and it can be characterized by the start (birth) and end (death) of the bar in $H_1$, as shown in Fig.~\ref{fig:TFIM_field_barcodes} (d).

Figure~\ref{fig:ising_H1_b_d} shows the birth and death value of $H_1$ at $h=1$ for different chain lengths. This model is exactly solvable (see Appendix~\ref{sec:exact_MI_TFIM}), and we can therefore compare the results to the infinite chain. As described earlier, the cycle is formed when the nearest neighbors connect, and this is set by the distance $D_{01}$. The hole closes when maximally separated spins connect, which in the case of an infinite chain is set by $D_{0\infty}$. For the infinite chain, the distances are $D_{01}/2\ln{2}\approx 0.744293$ and $D_{0\infty}/2\ln{2}= 1$, which matches the birth and death curves shown in Fig.~\ref{fig:ising_H1_b_d}.

\section{Example 2: XXZ spin chain in transverse field}
\label{sec:ex_2}
In order to study multiscale entanglement, we now focus on the 1D XXZ spin chain in tranverse field with Hamiltonian
\begin{equation}
    H=-\sum_i\left[ \left(\sigma_i^+\sigma_{i+1}^-+\sigma_i^-\sigma_{i+1}^+\right) + \frac{\Delta}{2}\sigma_i^z\sigma_{i+1}^z+h_i\sigma_i^z\right],
\end{equation}
where $\sigma_i^\pm = \sigma_i^x \pm i \sigma^y_i$, $\Delta$ is the interaction strength, and $h_i$ is a random magnetization. We take $\Delta=1$ and the field uniformly random $h_i=[-W,W]$. This model is often used to study the nature of the many-body localization (MBL) phase transition \cite{nidari2008,Herviou2019}. The critical value of disorder reported in the literature varies, usually close to $W_c\approx 3.8$ as obtained through large-scale ($L=26$) exact diagonalization \cite{Pietracaprina2018}. However, more recent studies suggest a slightly larger value closer to $W_c\approx 5$ \cite{PhysRevLett.115.187201,PhysRevB.98.174202}. In the ergodic phase (small $W$), the excited states have volume law entanglement. In the MBL phase (large $W$), the spins are mostly weakly entangled, but the entanglement structure is the object of our interest. Clusters (of spins) are defined as a subsystem that has stronger entanglement internally than with the rest of the system.

\begin{figure}[t]
    \centering
    \includegraphics{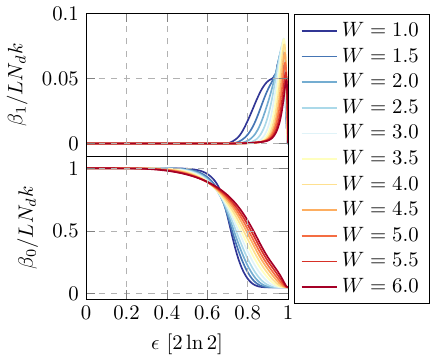}
    \caption{Betti numbers $\beta_0$ and $\beta_1$ for varying disorder strength $W$. XXZ spin chain with $L=16$ sites, $N_d=3000$ realizations and $k=50$ states.}
    \label{fig:xxz_betti}
\end{figure}

Using exact diagonalization, we select $k=50$ eigenstates from the middle of the spectrum for $N_d=3000$ disorder realizations (i.e. changing the random magnetization $h_i$). For each eigenstate, we construct a Vietoris-Rips complex with the distance metric of Equation~\ref{eq:MI_distance}. The resulting persistence barcodes are merged by discarding the information from which realization each barcode originated. Since this leads to barcodes with many bars, we use Betti numbers (Equation~\ref{eq:betti_number}) to summarize and investigate the results. Figure~\ref{fig:xxz_betti} shows the Betti numbers for varying disorder strengths $W$. As the disorder strength $W$ decreases, there are fewer connected components as counted by the normalized $\beta_0$.

Figure~\ref{fig:xxz_max_betti_1} shows the maximum of $\beta_1$ (see Fig.~\ref{fig:xxz_betti}) for varying disorder strength $W$. This maxima occur at a length scale around $\epsilon\approx 0.95 \cdot 2\ln{2}$. The Betti number $\beta_1$ essentially counts the number of spin clusters that are more strongly entangled (strength set by $\epsilon$) within the cluster than the environment. The minimum set of spins to form a 1-cycle is four spins (because three connected spins would form a 2-simplex in a Vietoris-Rips complex). For smaller system sizes $L\leq 14$, the maximum is at small disorder strength $W$. However, for larger systems, a maximum in the expected range of $3-5$ appears. Extrapolating the peak leads to a predicted critical strength of $W_c\approx 4.79$, which is in line with the aforementioned values reported in literature.

\begin{figure}[t]
    \centering
    \includegraphics{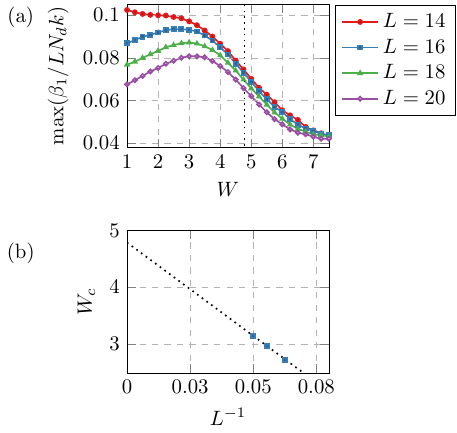}
    \caption{(a) Maximum Betti number of the first homology group for varying disorder strength $W$. 3000 disorder realizations (b) The disorder strength with largest Betti number, i.e. peak in (a), that extrapolates to a critical strength of $W_c\approx 4.79$ (marked by dotted line).}
    \label{fig:xxz_max_betti_1}
\end{figure}

\section{Discussion}
\label{sec:discussion}
Here we lay out a few examples of related work and future directions to explore. Computing the entanglement entropy of all possible partitions is numerically impractical due to the factorial number of possibilities. For single pairs of $N$ particles (or spins), the number is ``$N$ choose 2'' (scales as $O(N^2)$). However, due to the non-additivity of entropy, this does not necessarily capture all the structure present in the entanglement of the state.

Rather than starting from individual spins, it is also possible to perform a recursive  bipartitioning of the spin chain (as used in \cite{Herviou2019}). This turns the wave function into a binary tree, where each node is a set of spins. Compared to our method, which is bottom-up from spin pairs, this could be considered a top-down approach. Given the binary tree, the distance between leaf nodes (single spins) can be used to compute homology. This would be an alternative to our bottom-up approach based on mutual information and an interesting direction to explore in the future.

Another interesting future direction is to connect barcodes to matrix product states and tensor networks. In matrix product states, the bond dimension refers to the dimension of matrices used to represent the quantum state, and a small bond dimension comes with low computational complexity. For example, ground states of gapped Hamiltonians obey the area law and have constant bond dimension \cite{Hastings2007}. The bond dimension is dependent on the decay of correlations in the model, and these correlations are measured by the barcode, providing a potential bridge between the two. Regarding tensor networks, the Multiscale Entanglement Renormalizaton Ansatz (MERA) \cite{PhysRevLett.101.110501} is used for infinite 1D quantum critical states and represents an interesting case to study through the lens of homology. Another state to consider is the Rainbow State, which is a ground state of the inhomogeneous Ising transverse field chain that also exhibits a volume law \cite{Ramirez2015-lr,PhysRevB.104.195147}.

We compute entanglement entropy using singular value decomposition (SVD), but diagonalizing the partial trace of a density matrix (constructed as the outer product of the pure state) has similar computational complexity. The purity (or linear entropy) of a quantum state could be an alternative to the von Neumann entropy that is used here, and it has the benefit of not requiring diagonalization of the density matrix.

The first example of the Ising chain discussed in this paper can be solved efficiently using the Jordan-Wigner transformation. However, this is a special case, and it is not generally efficient for all interacting Hamiltonians, so we do not rely on this transformation in this work. Furthermore, the periodic boundary conditions are essential for forming a 1-cycle in the transverse-field Ising chain and this appears to limit the usefulness of barcodes. In Appendix~\ref{sec:TFIM_OBC} we show an example with open boundary conditions, where a phase transition is still captured.

The von Neumann entropy is typically not experimentally accessible, whereas the R\'enyi entropy of order $\alpha=2$ is experimentally accessible in certain setups \cite{Brydges2019}. This would make it possible to compute the barcode of an experimental setup.

Regarding persistent homology, we note that our homology groups are computed with $\mathbb{Z}_2$ (boolean) coefficients, the most common choice in topological data analysis. The choice of coefficients determines whether a simplex can occur multiple times in a chain or whether the direction of travel (represented by the sign of the coefficient) matters. The change in barcodes for entanglement structure when computing homology with different coefficients is also an interesting direction to explore. For examples of manifolds and the effect of homology with different coefficients, see Appendix~\ref{sec:non_orientable_torsion}.

TDA is a powerful tool that can be applied to the questions of emergent structures in the topology of the Hilbert space describing the system. This connects directly to a deeper question: does spacetime emerge from entanglement? This is one of the questions that is part of the Simons Collaboration: \emph{It from Qubit} \cite{ITFROMQUBIT}. The study of the emergent geometry from entanglement is still a topic of active research \cite{Raamsdonk2010,Roy2020,Carroll2021aiq,NeyForthcoming-NEYFQE}. For example, Cao \emph{et al.} have proposed the use of mutual information and classical multidimensional scaling to form emergent geometry \cite{Cao2017}.

\section{Conclusion}
\label{sec:conclusion}
Persistent homology is a novel tool for dealing systematically with the entanglement structure in quantum states. The topological features that it computes can be used to answer a number of scientific questions. First, the barcode of a quantum state changes dramatically when the state undergoes a quantum phase transition. It is therefore a new type of quantum order parameter. We have demonstrated its capability for two basic examples of the Ising chain and XXZ spin chain in a transverse field. Second, it can guide the development of suitable wave function ansatze, with for example tensor networks. Third, it provides a numerical approach for studying the emergent geometry of entanglement. This is very relevant to the scientific question whether spacetime emerges from entanglement.

The focus in this multidisciplinary work is on the application of persistent homology to equilibrium quantum phase transitions, the first of the aforementioned uses of the method. However, it is also possible to examine many quantum phenomena through the lens of homology and Betti numbers of the entanglement structure. The effect of adding time dependence to the presented scheme is also an promising future direction to explore, and recent work has already introduced some persistent homology observables for quantum many-body dynamics\cite{Spitz2021}.

\section{Acknowledgements}\label{acknowledgements}
This work emerged from initial discussions and research with Alexander V. Balatsky, for which the author thankfully acknowledges. The author also acknowledges fruitful exchanges of ideas with Qian Yang,  Lo\"{i}c Herviou, Jens H. Bardarson, R. Matthias Geilhufe, Gerben Oling and Benjo Fraser. This research is funded by the VILLUM FONDEN via the Centre of Excellence for Dirac Materials (Grant No. 11744), the European Research Council ERC HERO-810451 grant, University of Connecticut, and the Swedish Research Council (VR) through a neutron project grant (BIFROST, Dnr. 2016-06955). Nordita is partially supported by Nordforsk. The author also acknowledge computational resources from the Swedish National Infrastructure for Computing (SNIC) at the High Performance Computing Centre North (HPC2N).

\bibliography{references}

\appendix
\section{Bipartite entanglement entropy}
The distance metric used in our work is based on quantum mutual information and the properties of entanglement entropy. We therefore recall the key properties of entanglement entropy. Bipartite entanglement entropy is computed by performing Schmidt decomposition on a quantum state. We start by breaking down the Hilbert space in two parts $\mathcal{H}=\mathcal{H}_A\otimes \mathcal{H}_B$, and form a matrix $C_{ij}$ such that,
\begin{equation}
    \ket{\psi}=\sum_{i=1} \sum_{j=1} C_{ij} \ket{i}_A \otimes \ket{j}_B,
\end{equation}
where $\{\ket{i}_A\}$ and $\{\ket{j}_B\}$ are the basis sets for $\mathcal{H}_A$ and $\mathcal{H}_B$, respectively \cite{preskill1998lecture}.
Schmidt decomposition is essentially singular value decomposition $\bm{C}=\bm{U}\bm{\Sigma}\bm{V}^\dagger$ where $\lambda_\alpha$ forms the diagonal of $\bm{\Sigma}$. Assuming $N_B > N_A$, then the state can be expressed in the following form:
\begin{equation}
    \ket{\psi}=\sum_\alpha^{2^{N_A}} \lambda_\alpha \ket{\alpha}_A\otimes \ket{\alpha}_B
\end{equation}
If there are more than one non-zero singular values $\lambda_\alpha$, then the state is entangled. The singular values can also be used to compute the entanglement entropy
\begin{equation}
    S = -\sum_{j} |\lambda_j|^2 \ln{\left(|\lambda_j|^2\right)}
\end{equation}
which is zero if the state $\ket{\psi}$ is a product state (without entanglement). It is also upper bounded to $\ln{\left(d\right)}$ where $d$ is the Hilbert space dimension. This is $N\ln{2}$ in case of $N$ spins. Two other properties of quantum entropy that are used to form the distance metric are the triangle inequality,
\begin{equation}
    |S_i-S_j|\leq S_{ij} \leq S_i+S_j,
    \label{eq:S_triangle_ineq}
\end{equation}
and strong subadditivity \cite{Lieb1973},
\begin{equation}
    S_x+S_z \leq S_{xy} + S_{yz}.
    \label{eq:S_strong_subadd}
\end{equation}

In the main text we show how the bipartite von Neumann entanglement entropy $S$ is used to compute mutual information between subsystems and how it sets the length scale for the corresponding persistence barcodes.

\section{Persistent homology algorithm}
\label{sec:ph_algorithm}
The computation of persistent homology for a sequence simplicial complexes is in practice simply a row reduction of a matrix. More specifically, the boundary matrix encodes the boundaries present in the complex and each row and column represent a simplex. The rows and columns are also ordered by the appearance time of the simplices in the complex, such that the first column/row appears before the next in the sequence of complexes, and so on. In other words, this ordering originates from the input distance metric (e.g. Euclidean distance), as is visualized in Figure~\ref{fig:barcode_example}.

The \emph{standard algorithm} (sometimes called the \emph{column algorithm}) reduces the boundary matrix, see Algorithm 1 below \cite{books/daglib/0025666,10.5555/795666.796607}.
\begin{algorithm}[H]
  \caption{Column algorithm}
  \begin{algorithmic}[1]
      \For{$j=1$ to $m$}
        \While{there exists $j_0<j$ with $low(j_0)=low(j)$}
            \State add column $j_0$ to column $j$
        \EndWhile
      \EndFor
  \end{algorithmic}
  \label{alg:1}
 \end{algorithm}
The worst-case time complexity of this algorithm is cubic in the number of simplices, and a sparse matrix implementation is possible \cite{10.5555/795666.796607}. From the reduced boundary matrix we can read off the barcode. For other coefficients than $\mathbb{Z}_2$, see \cite{10.5555/3115476.3115792}.

We now demonstrate the algorithm on an example simplicial complex,
\begin{equation}
    K=\ \begin{gathered}\includegraphics[scale=.8]{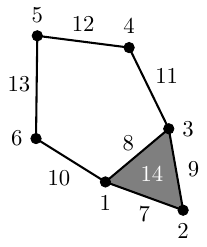}\end{gathered}
\end{equation}
This corresponds to a space of a disk (2-simplex number 14) glued to a 1-dimensional hole (cycle of 1-simplices, 8-11-12-13-10). First, deconstruct the simplicial complex into a boundary matrix,
\begin{equation}
    B=\ \begin{gathered}\includegraphics[scale=.6]{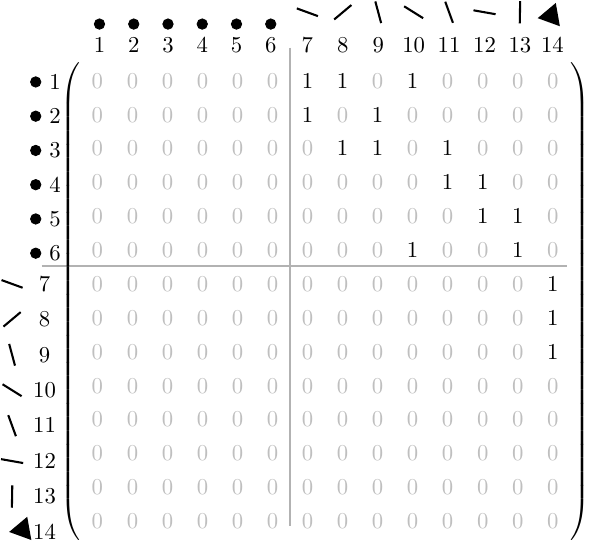}\end{gathered}
     \label{eq:full_boundary_matrix}
\end{equation}
where the non-zero rows indicate the boundary of the corresponding column. Although we did not specify the ordering when defining the simplicial complex, the simplices in the columns and rows are also ordered by their appearance time in the complex.

Next, reduce the boundary matrix following Algorithm~\ref{alg:1} and indicate the lowest $1$s with a box,
\begin{equation}
    B=\ \begin{gathered}\vspace{2cm}\includegraphics[scale=.6]{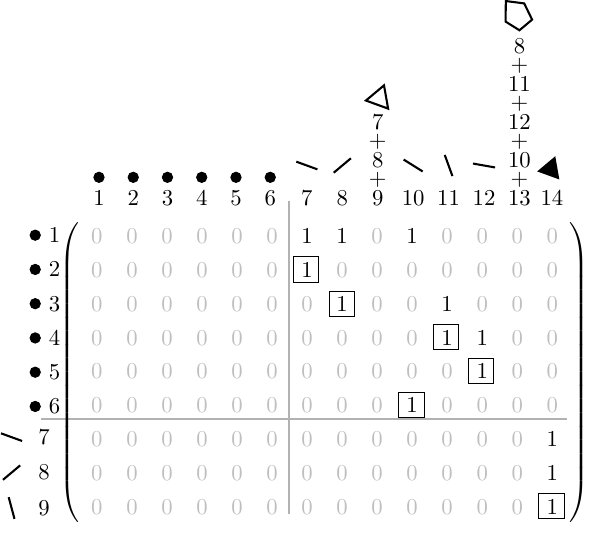}\end{gathered}
    \label{eq:reduced_boundary_matrix}
\end{equation}
The persistence pairs can be harvested from the reduced boundary matrix. An all-zero column indicates the \emph{birth} of a persistence pair. It \emph{dies} by pairing it with a lowest $1$. For example, the column $9$ is all-zero and initiates a pair. For the corresponding row $9$, it pairs with column 14, leading to persistence pair $[9,14)_1$, where the subscript indicates the birth simplex is $1$-dimensional. The all-zero columns that cannot be paired live until infinity. In our example, all the persistence pairs are:
\begin{align}
    [2,7)_0\\
    [3,8)_0\\
    [9,14)_1\\
    [6,10)_0\\
    [4,11)_0\\
    [5,12)_0\\
    [13,\infty)_1\\
    [1,\infty)_0
\end{align}
As expected, only one connected component (0 dimensionality hole) survives to infinity: $[1,\infty)_0$.

The persistence pairs are typically visualized in a barcode (see Fig.~\ref{fig:barcode_example}), where the birth and death indicate the start and end of the bar. Note that the explicit length scales associated to the appearance of each simplex is not part of the persistent homology algorithm, only the ordering is. The birth and death simplices can be translated into birth and death length scales when plotting a barcode by referring to the input sequence of simplicial complexes that was generated using a distance metric.

\section{Non-orientable surfaces and torsion}
\label{sec:non_orientable_torsion}
The entanglement structure of a quantum state is converted into a topological space using mutual information as a metric. This space is then studied using simplicial homology with $\mathbb{Z}_2$ coefficients, which means that simplices either are included or excluded when talking about loops in this space. This is computationally easier, however, the more general case of homology with integer coefficients (i.e. integral homology), is able to capture more information. In this section we study this in detail and show that manifolds may have different homology groups depending on the choice of coefficients.

Non-orientable surfaces are twisted and turn clockwise objects to counterclockwise, with common examples being the Möbius strip, Klein bottle and the real projective plane. A surface is orientable if it has a consistent notion of clockwise rotation as we move around. It turns out that non-orientability is connected to the notion of \emph{torsion}. Algebraically, torsion means that a group element has \emph{finite order} (i.e. $g^n=e$ for a positive integer $n$) and a group where all elements fit this condition is called a torsion group. A torsion-free group has no elements of finite order, other than the identity element, e.g. $(\mathbb{Z},+)$. Both orientable (e.g. $\mathbb{R}\mathbb{P}^3$) and non-orientable (e.g. $\mathbb{R}\mathbb{P}^2$) manifolds may have torsion in its integral homology. However, there is a theorem that states that if $M$ is a closed connected $n$-manifold that is non-orientable, the homology group $H_{n-1}$ contains torsion \cite[Corollary~3.28]{Hatcher2001-mw}.

In general, the $k$-homology group over the integers for a simplicial complex $K$ has the following form
\begin{equation}
    H_k(K,\mathbb{Z})=\overbrace{\mathbb{Z} \oplus \dots \oplus \mathbb{Z}}^{\beta_k} \oplus\ \mathbb{Z}_{k_1} \oplus \dots \oplus \mathbb{Z}_{k_n}
\end{equation}
where $\beta_k$ is the Betti number and $k_i$ are the torsion coefficients. However, it is common to compute homology over $\mathbb{Z}_2$ rather than the integers, leading to a different kind of homology, namely $H_k(K,\mathbb{Z}_2)$. To illustrate the absence or presence of torsion in non-orientable surfaces, we start with the Möbius strip and then proceed to the real projective plane.

\begin{figure}[t]
    \centering
    \includegraphics{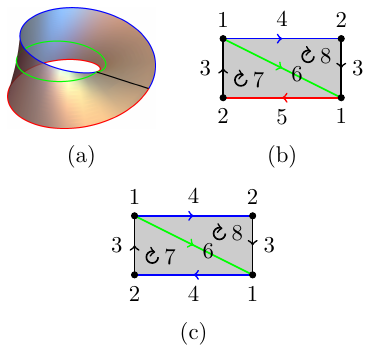}
    \caption{(a) Möbius strip that is triangulated by glueing together two 2-simplices. The faces are four unique 1-simplices and two 0-simplices. (b) Simplicial complex of the Möbius strip showing the orientation of the simplices. (c) Simplicial complex of the real projective plane that is constructed by glueing the sides of the Möbius strip together, this can be done in $\mathbb{R}^4$ without the surface intersecting itself.}
    \label{fig:moebius_rp2}
\end{figure}

Figure~\ref{fig:moebius_rp2} (a) and (b) shows the Möbius strip $M$ and its corresponding simplicial complex. We now compute its 1-homology over the integers. The chain groups for the Möbius strip $M$ are
\begin{align}
    C_0&=\left<[1],[2]\right>\\
    C_1&=\left<[3],[4],[5],[6]\right>\\
    C_2&=\left<[7],[8]\right>
\end{align}
where the simplices form a basis and the coefficients are hidden, e.g. $C_0=\{a_1[1]+a_2[2]\}$.
The $1$-cycles are a subgroup of $C_1$, namely
\begin{align}
    Z_1&=\{c \in C_1\ |\ \partial(c) = 0\}\\
    &=\{a_3[3]+a_4[4]+a_{5}[5]+a_{6}[6]\}
\end{align}
where the coefficients $a_i\in \mathbb{Z}$ are constrained by $\partial(c)=0$. In order to find the space of 1-cycles, we apply the boundary operator (Equation~\ref{eq:boundary_operator_Z}) to the 1-simplices,
\begin{align}
    \partial [3]&=[1]-[2]\\
    \partial [4]&=[2]-[1]\\
    \partial [5]&=[2]-[1]\\
    \partial [6]&=[1]-[1]=0
\end{align}
and row reduce this system of equations. The space of 1-cycles is generated by
\begin{equation}
    Z_1=\left<[3]+[4],[3]+[5],[6]\right>.
\end{equation}
Next, to find the boundary group $B_1=\text{Im}\ \partial_2=\partial_2(C_2)$ we apply the boundary operator to the 2-simplices:
\begin{align}
    \partial[7] &= [3]+[6]+[5]\\
    \partial[8] &= [3]-[6]+[4].
\end{align}
Therefore, the boundary group $B_1$ is generated by
\begin{align}
    B_1=\left<[3]+[6]+[5],[3]-[6]+[4]\right>
\end{align}
Finally, we find that the $1$-homology (cycles mod boundaries) over the integers for the Möbius strip $M$ is
\begin{align}
    H_1(M,\mathbb{Z})&=\frac{\left<[3]+[4],[3]+[5],[6]\right>}{\left<[3]+[6]+[5],[3]-[6]+[4]\right>}\\
    &\cong\frac{\left<[3]+[4]+[6],[3]+[5]-[6],[6]\right>}{\left<[3]+[6]+[5],[3]-[6]+[4]\right>}\\
    &\cong \left<[6]\right>\\
    &\cong \mathbb{Z}.
\end{align}
The number of one-dimensional holes is described by the Betti number $\beta_1=1$. Geometrically, we can see this loop in Fig.~\ref{fig:moebius_rp2} (a) in the middle of the Möbius strip, and observe that the group $H_1$ is torsion-free.

In the case of $\mathbb{Z}_2$ coefficients, the boundary operator is
\begin{equation}
    \partial S=\sum_{j=0}^k\left[v_0 v_1\dots \hat{v}_j \dots v_k\right]
\end{equation}
with $\hat{v}_j$ simplex being left out. Simplices appearing twice mod out to zero, and again $\partial^2 S=0$ is true. The Betti number $\beta_k$ is now defined as the number of copies of $\mathbb{Z}_2$. Performing the same steps are before we arrive at the 1-homology with $\mathbb{Z}_2$ coefficients for the Möbius strip,
\begin{align}
    H_1(M,\mathbb{Z}_2)&=\frac{\left<[3]+[4],[3]+[5],[6]\right>}{\left<[3]+[6]+[5],[3]+[6]+[4]\right>}\\
    &\cong\frac{\left<[3]+[4]+[6],[3]+[5]+[6],[6]\right>}{\left<[3]+[6]+[5],[3]+[6]+[4]\right>}\\
    &\cong \left<[6]\right>\\
    &\cong \mathbb{Z}_2 \label{eq:mobius_H_z2}
\end{align}
where the Betti number $\beta_k$ is now defined as the number of copies of $\mathbb{Z}_2$ (i.e. $\beta_1=1$ in this case).

The real projective plane $\mathbb{R}\mathbb{P}^2$ can be constructed by embedding the Möbius strip in $\mathbb{R}^4$ and glueing the sides together with the right orientation, see Fig.~\ref{fig:moebius_rp2} (b) and (c). The chain groups for this manifold are the same except for the 1-chains and 1-cycles
\begin{align}
    C_1&=\left<[3],[4],[6]\right>\\
    Z_1&=\{c \in C_1\ |\ \partial(c) = 0\}\\
    &=\{a_3[3]+a_4[4]+a_{6}[6]\}
\end{align}
Again, the coefficients are constrained by the boundary operator $\partial(c)=0$ and the space of the 1-cycles is generated by
\begin{equation}
    Z_1=\left<[3]+[4],[6]\right>.
\end{equation}
The boundary of the 2-simplices are now
\begin{align}
    \partial[7] &= [3]+[6]+[4]\\
    \partial[8] &= [3]-[6]+[4]
\end{align}
The boundary group $B_1$ is generated by
\begin{equation}
    B_1=\left<[3]+[6]+[4],[3]-[6]+[4]\right>
\end{equation}
The 1-homology group for the real projective plane
\begin{align}
    H_1(\mathbb{R}\mathbb{P}^2,\mathbb{Z})&=\frac{\left<[3]+[4],[6]\right>}{\left<[3]+[6]+[4],[3]-[6]+[4]\right>}\\
    &\cong \frac{\left<[3]+[4]+[6],[6]\right>}{\left<[3]+[6]+[4],2[6]\right>}\\
    &\cong \frac{\left<[6]\right>}{\left<2[6]\right>}\cong \frac{\mathbb{Z}}{2\mathbb{Z}}\equiv \mathbb{Z}_2
\end{align}
The real projective plane has 2-torsion and a Betti number $\beta_1=0$. Geometrically, it means that the single move from $[1]$ to $[1]$ is a loop, but it cannot be contracted to a point, however, twice the loop (i.e. walking around the rectangular perimeter of the simplicial complex in Fig.~\ref{fig:moebius_rp2}) is contractible to zero.

In the case of $\mathbb{Z}_2$ coefficients, the real projective plane has a 1-homology of
\begin{align}
    H_1(\mathbb{R}\mathbb{P}^2,\mathbb{Z}_2)&=\frac{\left<[3]+[4],[6]\right>}{\left<[3]+[6]+[4],[3]+[6]+[4]\right>}\\
    &\cong \frac{\left<[3]+[4]+[6],[6]\right>}{\left<[3]+[6]+[4],2[6]\right>}\\
    &\cong \left<[6]\right>\cong \mathbb{Z}_2.
\end{align}
This is the same result as for the Möbius strip (Equation~\ref{eq:mobius_H_z2}), with Betti number $\beta_1=1$, so it is clear that homology over $\mathbb{Z}_2$ coefficients cannot distinguish between these two topological spaces. In contrast, the integral homology is different between the two spaces due to the presence of torsion in the real projective plane. However, it is also important to note that orientable manifolds (e.g. $\mathbb{R}\mathbb{P}^3$) can have torsion in its integral homology groups.

More generally, there is the \emph{universal coefficient theorem} that relates the integral homology to homology with different coefficients $A$, and it turns out that the integral homology completely determines the homology groups for any abelian group $A$ \cite[Section 3.A]{Hatcher2001-mw}.

Our study of entanglement structures only considers homology with coefficients $\mathbb{Z}_2$ since it is the most computationally straightforward. However, as shown here, integral homology groups are sometimes more informative. There exist algorithms that efficiently compute persistent homology with different coefficients and deduce the torsion subgroups of the integral homology \cite{Boissonnat_2019}. Whether torsion provides important information in the study of entanglement structures is unclear and the topic of future work.

\section{Transverse-field Ising chain: entanglement entropy at criticality}
\label{sec:exact_MI_TFIM}
The demonstration of persistent homology on the one-dimensional transverse-field Ising chain (with periodic boundary conditions) relied on the general approach of exact diagonalization (ED), and this scales to finite-size systems up to about $L=30$ sites. However, this model is exactly solvable and this makes it possible to compute the entanglement entropy where a finite number of sites are in subsystem A while the rest of the infinitely long chain is in subsystem B. This also provides exact results for the entanglement entropy and the mutual information used in our work.

Persistent homology captures the structure of the entanglement as described by the two-site mutual information $M_{ij}$. Translational invariance means that the mutual information only depends on the distance between two sites, i.e. $M_{ij}=M_{0r}$ with $r=|i-j|$. The distance metric used in persistent homology is simply the inverse of the mutual information, such that high mutual information means two sites are close, see Eq.~\ref{eq:MI_distance}.

At the critical point, the asymptotic behaviour of the persistent homology for this model is governed by two special cases. The first case sets the birth length scale of the 1-homology, when neighboring spins are connected and form a cycle (Fig.\ref{fig:TFIM_field_barcodes}(b)). This is set by the amount of mutual information between neighboring spins, i.e. $M_{01}$. The second case is the point when the 1-hole is closed and it sets the death length scale. This is set by the amount of mutual information between spins on opposite sides of the chain, which in a case of an infinite chain corresponds to $M_{0\infty}$. Therefore, we need to compute the entanglement entropies $S_{1}$, $S_{01}$ and $S_{0\infty}$ (see Eq.~\ref{eq:mutual_information}). The exact one-site and two-site reduced density matrix and the spin-spin correlation functions for the infinite transverse-field Ising chain (at criticality) are known \cite{PFEUTY197079,PhysRevA.3.786,PhysRevA.66.032110,720404}.

Starting with the case of the single site in subsystem A, the reduced density matrix is
\begin{align}
    \rho_1=\frac{I+\left<\sigma^z\right>\sigma^z}{2},
\end{align}
where $I$ and $\sigma^z$ are the identity matrix and the Z Pauli matrix, respectively. For the ground state, the transverse magnetization is \cite{PhysRevA.66.032110} 
\begin{align}
    \left< \sigma^z \right> = \frac{1}{\pi} \int_0^\pi \mathrm{d}\phi \frac{1+\lambda\cos \phi }{\sqrt{1+\lambda^2+2\lambda \cos\phi}}=\frac{2}{\pi}.
\end{align}
This gives the exact entanglement entropy,
\begin{align}
    \frac{S_1}{\ln{2}} &= \frac{-\text{Tr} \left( \rho_1 \ln \rho_1 \right)}{\ln{2}}\\
    &= \frac{2\pi\ln{\left(\frac{4\pi^2}{\pi^2-4}\right)}-8\tanh ^{-1}\left(\frac{2}{\pi }\right)}{\pi  \ln (16)}\\
    &\approx 0.68376,
\end{align}
normalized by $\ln{(2)}$ such that it ranges from 0 to 1.

The two-site case takes two spins $i$ and $j$ in subsystem A,
\begin{align}
\rho_{0r} =\frac{I_{0r}+\left<\sigma^z\right>\left(\sigma_0^z+\sigma_r^z\right)+\sum_{k=1}^3\left<\sigma_0^k\sigma_r^k\right>\sigma_0^k\sigma_r^k}{4},
\end{align}
which can be described with $r=|i-j|$ due to translational symmetry. The spin-spin correlation functions are known at criticality,
\begin{align}
    \left<\sigma_0^x\sigma_r^x\right>&=\left(\frac{2}{\pi}\right)^r 2^{2r(r-1)}\frac{H(r)^4}{H(2r)}\\
    \left<\sigma_0^y\sigma_r^y\right>&=-\frac{\left<\sigma_0^x\sigma_r^x\right>}{4r^2-1}\\
    \left<\sigma_0^z\sigma_r^z\right>&=\left<\sigma_0^z\right>\left<\sigma_r^z\right>+\frac{4}{\pi^2}\frac{1}{4r^2-1},
\end{align}
with $H(r)=1^{r-1}2^{r-2}\dots (r-1)$. The two-site entanglement entropy for neighboring sites is
\begin{align}
    \frac{S_{01}}{2\ln{2}} &= \frac{-\text{Tr} \left( \rho_{01} \ln \rho_{01} \right)}{2\ln{2}}\\
    &\approx 0.42805.
\end{align}

For the two-site case with limit $r\rightarrow\infty$, the correlation functions are $\left<\sigma_0^x\sigma_r^x\right>=0$, $\left<\sigma_0^y\sigma_r^y\right>=0$ and $\left<\sigma_0^z\sigma_r^z\right>=\left<\sigma_0^z\right>\left<\sigma_r^z\right>=4/\pi^2$, and the entanglement entropy,
\begin{align}
    \frac{S_{0\infty}}{2\ln{2}} &= \frac{-\text{Tr} \left( \rho_{0\infty} \ln \rho_{0\infty} \right)}{2\ln{2}}\\
    &=\frac{S_1}{\ln{2}}\\
    &\approx 0.68376.
\end{align}
For completeness, Fig.~\ref{fig:S_01_TFIM} shows the entanglement entropy as a function of $r$. This shows that neighboring sites with $r=1$ have the lowest entanglement entropy and this leads to the smallest mutual information distance. This matches what we see for a small chain with $h=1$ in Fig.~\ref{fig:1d_qmi_tfim} (b).

\begin{figure}[t]
    \centering
    \includegraphics{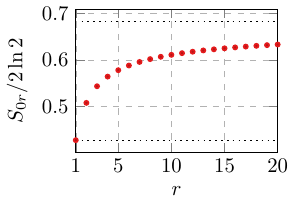}
    \caption{Exact two-site entanglement entropy $S_{0r}/2\ln{2}$ (where $r=|i-j|$) for an infinitely long transverse-field Ising chain with periodic boundary conditions. The dotted lines indicate the minimum $S_{01}$ and the maximum $S_{0\infty}$.}
    \label{fig:S_01_TFIM}
\end{figure}

Finally, we find that the mutual information (Eq.~\ref{eq:mutual_information}) for the case of two neighboring spins is
\begin{align}
    \frac{M_{01}}{2\ln{2}}=\frac{S_0+S_1-S_{01}}{2\ln{2}}\approx 0.255707.
\end{align}
Therefore the distance (Eq.~\ref{eq:MI_distance}) $D_{01}=1-M_{01}\approx 0.744293$, which matches the birth length scale shown in Fig.~\ref{fig:ising_H1_b_d}. When the spins are far apart, the mutual information is zero,
\begin{align}
    \frac{M_{0\infty}}{2\ln{2}}=\frac{S_0+S_1-S_{0\infty}}{2\ln{2}}=0,
\end{align}
meaning that the distance is the maximum $D_{0\infty}=2\ln{2}$. This refers to the death length scale Fig.~\ref{fig:ising_H1_b_d}. In other words, the birth and death of the 1-homology of this model at criticality converge to their expected values.

\section{Transverse-field Ising chain: persistent homology with open boundary conditions}
\label{sec:TFIM_OBC}

\begin{figure}[b]
    \centering
    \includegraphics{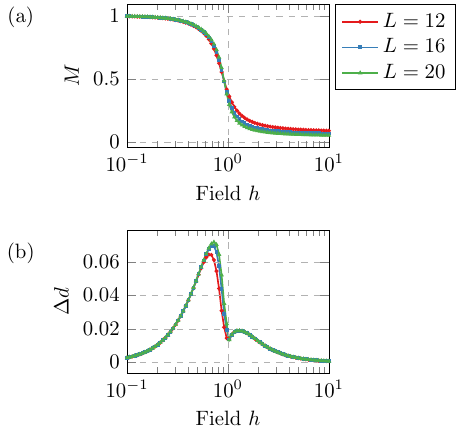}
    \caption{The ground state magnetization $M=|1/N \sum_{i}\sigma_i^x|$ of the transverse-field Ising chain with open boundary conditions (a). Range of the set of 0-homology deaths, where the maximum range coincides with the transition close to $h=1$.}
    \label{fig:TFIM_OBC}
\end{figure}

The entanglement structure of a quantum state and its corresponding (persistence) barcode depend on the Hamiltonian and its boundary conditions. For example, the existence of the 1-cycle for the tranverse-field Ising chain at criticality (Fig.~\ref{fig:TFIM_field_barcodes} (b)) relies on periodic boundary conditions (PBC) being present. This raises the question whether persistent homology is able to detect transitions without PBC. Here we show that the model with open boundary conditions also shows structure in its persistence barcode. However, it is the 0-homology instead of the 1-homology that indicates a phase transition.

Without translational invariance due to PBC, the mutual information between pairs of neighboring spins is not always exactly the same. For example, close to the critical transverse-field of $h=1$, we see that some pairs are closer as defined by the mutual information distance metric (Eq.~\ref{eq:MI_distance}), which means that the spins connect and form 1-simplices earlier. In terms of 0-homology (connected components), this coincides with the end (typically referred to as death) of a bar, e.g. Fig.~\ref{fig:barcode_example}. Note that 0-homology differs from 1-homology in that the bars always start (typically referred to as birth) at zero, indicating a disconnected data point (spin in our case). Therefore the only information is the set of death values. Since the 0-homology bars are all born at zero, this is equivalent to the set of lifetimes. The changes in 0-homology can be summarized in many ways, and here we plot the difference between the maximum and minimum of the set of death values Fig.~\ref{fig:TFIM_OBC} (b), and this indicates the transition around $h=1$ at its peak, similarly to the traditional magnetization (Fig.~\ref{fig:TFIM_OBC} (a)). This example shows that persistent homology is a general tool that does not necessarily rely on PBC being present.

\end{document}